\RequirePackage{amsmath}
\documentclass{article}
\usepackage[utf8]{inputenc}
\usepackage{amsfonts}
\usepackage{amsmath}
\usepackage{amssymb, bm}
\usepackage{MnSymbol}
\usepackage{stmaryrd}
\usepackage{amssymb}
\usepackage{mathtools, bm}
\usepackage{algorithm}
\usepackage[noend]{algorithmic}
\usepackage{color}                      
\usepackage{hyperref}
\usepackage{graphicx}
\usepackage{color}

\usepackage[title]{appendix}

\newcommand{\floor}[1]{\left\lfloor #1\right\rfloor}

\newcommand{\ZZ}{\mathbb{Z}}

\DeclareMathOperator{\conv}{conv}

\usepackage[affil-it]{authblk}
\usepackage{amsthm}
\newtheorem{theorem}{Theorem}
\newtheorem{lemma}[theorem]{Lemma}

\title{Digital Convex + Unimodular Mapping = 8-Connected (All Points but One 4-Connected)}
\author{Loïc Crombez}
\affil{Universit\'e Clermont Auvergne and LIMOS}

\begin{document}
\maketitle
\begin{abstract}

In two dimensional digital geometry, two lattice points are \emph{4-connected} (resp. \emph{8-connected}) if their Euclidean distance is at most one (resp. $\sqrt{2}$). A set $S \subset \ZZ^2$ is \emph{4-connected} (resp. \emph{8-connected}) if for all pair of points $p_1, p_2$ in $S$ there is a path connecting $p_1$ to $p_2$ such that every edge consists of a 4-connected (resp. 8-connected) pair of points.
The original definition of digital convexity which states that a set $S \subset \ZZ^d$ is \emph{digital convex} if $\conv(S) \cap \ZZ^d= S$, where $\conv(S)$ denotes the convex hull of $S$ does not
guarantee connectivity. However, multiple algorithms assume connectivity.
In this paper, we show that in two dimensional space, any digital convex set $S$ of $n$ points is unimodularly equivalent to a 8-connected digital convex set $C$.
In fact, the resulting digital convex set $C$ is 4-connected except for at most one point which is 8-connected to the rest of the set. 
The matrix of $SL_2(\ZZ)$ defining the affine isomorphism of $\ZZ^2$ between the two unimodularly equivalent lattice polytopes $S$ and $C$ can be computed in roughly $O(n)$ time.
We also show that no similar result is possible in higher dimension.

\end{abstract}

\section{Introduction}

Digital Geometry studies the geometry of \emph{lattice points}, those are the points with integer coordinates~\cite{KlR04}.
Convexity, a fundamental concept in continuous geometry~\cite{Ro89}, is naturally also fundamental in digital geometry. However, unlike in any linear space where convexity is clearly defined, several definitions have been investigated for convexity in digital geometry~\cite{KR82,KR82-2,Cha83,Ki96,ChR98}. Just in two dimensions, we encounter several definitions such as triangle line~\cite{KR82}, HV convexity~\cite{BDNP96}, Q convexity~\cite{Da01}.

Some of those definitions where created in order to guarantee that under those definitions a convex object is connected (in terms of the induced grid subgraph). No such guarantee is given by the original following definition of digital convexity which is the one that will be used throughout this paper: A set $S$ of $n$ lattice points is said to be \emph{digital convex} if $\conv(S) \cap \ZZ^d = S$, where $\conv(S)$ is the convex hull of $S$.
This definition is equivalent to saying that there exist a convex polyhedron $P$ in $\mathcal{R}^d$ such that $P \cap \ZZ^d = S$.
However, this definition of digital convexity provides a lot of mathematical properties such as being preserved under unimodular affine transformations $SL_2(\ZZ)$. These transformations are the lattice preserving mappings that also preserves parallel lines and area~\cite{HaN12,BaP92}.

In this paper we prove that any digital convex sets $S$ is unimodularly equivalent to an \emph{almost 4-connected} set $C$. We say that a set $C$ is \emph{almost 4-connected} if $C$ is 4-connected except for at most one point which is 8-connected to the rest of the set. We also propose an algorithm that computes such a set $C$ in roughly $O(n)$ time.

The demonstration of existence of such a set, and the algorithm are both based on the same technique which consists in mapping a \emph{lattice diameter}~\cite{BaF01} of $S$ to a horizontal line. 
A \emph{lattice diameter} of digital convex set $S$ is the longest string of integer points on any line in the Euclidean space that is contained in $\conv(S)$.
The lattice diameter is invariant under the group of unimodular affine transformations $SL_2(\ZZ)$. 

The construction consists in computing a lattice diameter $d$ of $S$ , applying an affine isomorphism $\ZZ^2$ that maps $S$ to a horizontal line, and finally adjusts through \emph{horizontal shear mapping} in order to obtain an almost 4-connected set. A horizontal shear mapping is a unimodular affine transformation that preserves the $y$-coordinates. The matrix of defining those transformations are of the form $
\begin{bmatrix}
1 & k\\
0 & 1
\end{bmatrix}
$.\\
Note that in 3 dimension the tetrahedron $((0,0,0), (0,1,0), (1,0,0), (1,1,k))$ has no lattice points in its interior and has a volume of $\frac{k}{6}$. As a consequence, since unimodular affine transformations preserve volume, similar results as the one presented in this paper for unimodular affine transformations are impossible in dimension higher than 2.

\section{Unimodularly equivalence to connected set} 
 \label{s:eq} 

The purpose of this section is to provide a constructive proof of the following theorem. 

\begin{theorem} \label{th:1}
    For any digital convex set $S$ there is a unimodular affine transformations that maps $S$ to an almost 4-connected set $C$.
\end{theorem}

In order to prove theorem~\ref{th:1} we describe a multiple step construction that results in an almost 4-connected set. The steps of this constructions are the following:

\begin{itemize}
    \item First, we find a lattice diameter $d$ of $S$. We define $k$ as the number of lattice points located on $d$.
    \item We then map $S$ to $S_1$ using a unimodular affine transformation that maps $d$ to $d_1$ such that $d_1$ starts from the point $(0,0)$ and end at the point $(k-1,0)$. Note that $d_1$ is a lattice diameter of $S_1$.
    \item In Section~\ref{ss::reduc} we reduce the problem to the study of a subset $\meddiamond \in S_1$ such that the convex hull of $\meddiamond$ is the convex hull of: $d_1$, the topmost point in $S_1$, and the bottommost point in $S_1$.
    From here on in, for simplicity, we assume that the lattice point $p \meddiamond$ that is the furthest from the line $y=0$ is the topmost point in $S_1$.
    \item In Section~\ref{ss:top_part} we only consider $\medtriangleup \in \meddiamond$, the subset of lattice points above $d_1$. We show that there is a horizontal shear mapping $\mathcal{M_1}$ (a unimodular affine that does not affect $y$-coordinates) such that $\mathcal{M_1}(\medtriangleup)$ is 4-connected, to the exception of one special case that is treated in Section~\ref{ss:pom_pom}.
    \item Section~\ref{ss:bottom_part} focuses $\medtriangledown \in \meddiamond$, the subset of lattice points below $d_1$ and show that there is a shear mapping $\mathcal{M_2}$ such that $\mathcal{M_2}(\meddiamond)$ is almost 4-connected.
    \item Finally, in Section~\ref{ss:general_case} we go back to the general case and explain why our construction not only almost 4-connects $\meddiamond$, but also $S$.
\end{itemize}

\subsection{Reduction to a quadrilateral} \label{ss::reduc}
 
For simplicity, in most of this proof, we will consider the subset $\meddiamond \in S_1$ that is the intersection of the lattice grid $\ZZ^2$ with the convex hull of: $d_1$, $top$, and $bottom$, where $top$ and $bottom$ are respectively the topmost and bottommost points of $S_1$ (See Figure~\ref{fig:small_quad} a). This reduction is mostly possible thanks to Lemma~\ref{le:quadri}, even though we will have to take additional precautions detailed in Section~\ref{ss:general_case} when $\meddiamond$ is almost 4-connected but not 4-connected.

 \begin{figure}[tb] 
    \begin{center}
        \includegraphics[scale=0.7]{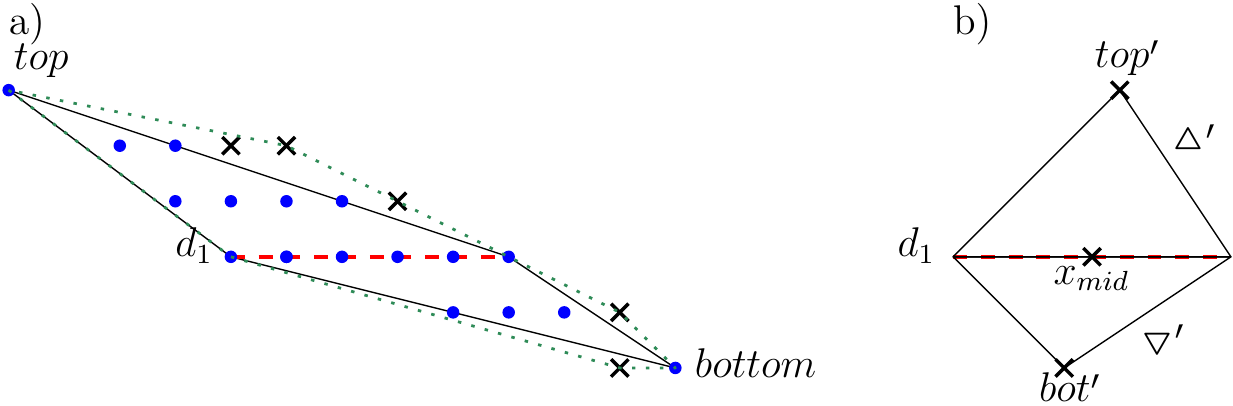}
    \end{center}
    \caption{a) The lattice diameter $d_1$ of $S_1$ is represented by the dashed red line segment. The subset $\meddiamond$ that is the union of an upper and lower triangle with the diameter as a common horizontal edge is represented by the blue dots. The set $S_1$ is the union of the blue dots, and the black crosses. b) Representation of $\meddiamond'$ after horizontal shearing. $x_{mid}$ is the middle of $d_1$ the lattice diameter of $\meddiamond$. $\medtriangleup$ and $\medtriangledown$ are the triangles on each side of $d_1$. The points $top'$ and $bot'$ are the two points with the most extreme $y$-coordinates.}
    \centering
    \label{fig:small_quad}
\end{figure}

\begin{lemma}\label{le:quadri}
    For any digital convex set $S_1$ such that $top$ (resp. $bottom$) are one of the topmost (resp. bottommost points) in $S_1$, and such that $\meddiamond \in S_1$ is the intersection of $\ZZ^2$ with the convex hull of the horizontal diameter $d_1$ and the two points $top$, $bottom$.
    For any horizontal shear mapping $\mathcal{M}$, if $\mathcal{M}(\meddiamond)$ is 4-connected, then $\mathcal{M}(S_1)$ is also 4-connected.
\end{lemma}
\begin{proof}
     We denote $y_t$ and $y_b$ the $y$-coordinates of the topmost and bottommost point in $\meddiamond$.
    As $\mathcal{M}$ is a horizontal shear mapping, it preserves $y$-coordinates, and as $\meddiamond_c$ is 4-connected, for each integer $y_i$ such that $y_t \geq y_i \geq y_b$ there is a point in $\meddiamond_c$ whose $y$-coordinate is equal to $y_i$.
    Now, we consider a point $p(x_p,y_p)$ in $S_c \notin \meddiamond_c$. We have $y_t \geq y_p \geq y_b$, and hence there is a point $p_q \in \meddiamond_c$ whose $y$-coordinate is $y_p$. As $S_c$ is digital convex, all lattice points on the horizontal line segment from $p$ to $p_q$ are in $S_c$, hence $p$ is 4-connected to $\meddiamond_c$.
\end{proof}

We denote $k$ the number of lattice points in $d$, the lattice diameter of $S_1$, which we mapped to a horizontal line segment $d_1$.
The leftmost and rightmost lattice points in $d_1$ are $p_0$ and $p_{k-1}$.
For simplicity, we choose the coordinates such that $p_0$ is the origin. Hence the coordinates of $p_0$ and $p_{k-1}$ are $p_0(0,0)$ and $p_{k-1}(k-1,0)$.
For simplicity, from now on, we assume that the $y$-coordinate of $top$ is larger or equal than the absolute value of the $y$-coordinate of $bottom$. The opposite case being symmetrically equivalent.

\subsection{Connecting the top} \label{ss:top_part}

In this section, we will only consider $\medtriangleup$, the top part of $\meddiamond$, consisting of all the points above $d_1$, $d_1$ included.
We denote $x_{mid} = \frac{k-1}{2}$.
We now apply to $\medtriangleup$ the horizontal shear mapping $\mathcal{M}_h$ such that all lattice points on $d_1$ maps to itself and such that $top(x_t,y_t)$ is mapped as close as possible to $(x_{mid}, y_t)$. The image of top by $\mathcal{M}_h$is $top'(top'_x,y_t)$. 
For simplicity we assume that $top'_x \leq x_{mid}$. The case when $top'_x \leq x_{mid}$ is symmetrically equivalent.
We define $\medtriangleup'$ as the set to which $\mathcal{M}_h$ maps $\medtriangleup$ (See Figure~\ref{fig:small_quad} b).
We now assume that $\medtriangleup'$ is not 4-connected and study the possible location of $top'$.

We know that:
\begin{itemize}
    \item[i] $top'_x < 0$ or  $top'_x > k-1$. Otherwise $\medtriangleup'$  would be 4-connected since the vertical segment going from $top'$ to $d_1$ is in $\medtriangleup'$ (See Figure~\ref{fig:topmap_pos} a).
    \item[ii] $top'$ cannot be located both left of the line $x=0$  and above the line $y=k-\frac{k}{k-1}x$. Otherwise the $k+1$ points: $(0,0), (0,1), ... ,(0,k)$ would be in $\medtriangleup'$ and $d_1$ would not be a lattice diameter of $\medtriangleup'$ (See Figure~\ref{fig:topmap_pos} b).
    \item[iii] $top'$ is located above or on the line $y=\frac{k-1}{2}-2x$. Otherwise there would be a horizontal mapping that maps $top$ closer to the point $(x_{mid},y_t)$ (See Figure~\ref{fig:topmap_pos} c).
    \item[iv] $top'$ is not located both above the line $y=d-1-x$ and below the line $y=-dx$. Otherwise the $k+1$ points: $(k-1,0),(k-2,1),...(0,k-1),(-1,k)$ would be in $\medtriangleup'$ and $d_1$ would not be a lattice diameter (See Figure~\ref{fig:topmap_pos} d).
\end{itemize}
Hence, the only possible location for $top'$ is strictly inside the quadrilateral $Quad$ whose vertices are: $(0,k-1),(0,k), (-1,k),(-\frac{k-1}{k-2},k\frac{k-1}{k-2})$.
As $x=-1$ is the only vertical line of integer coordinates intersecting the inside of $Quad$, all lattice points strictly inside $Quad$ are located on the line $x = -1$.
The line $x=-1$ intersects the edges of $Quad$ at the points: $(-1,k)$ and $(-1, k + \frac{k}{k-1})$. Since $k\geq 2$, we have $k/(k-1)\leq 2$. Hence there is exactly one lattice point located strictly inside $Quad$: $(-1,k+1)$.(See Figure~\ref{fig:topmap_pos} e).

 \begin{figure}[tb] 
    \begin{center}
        \includegraphics[scale=0.7]{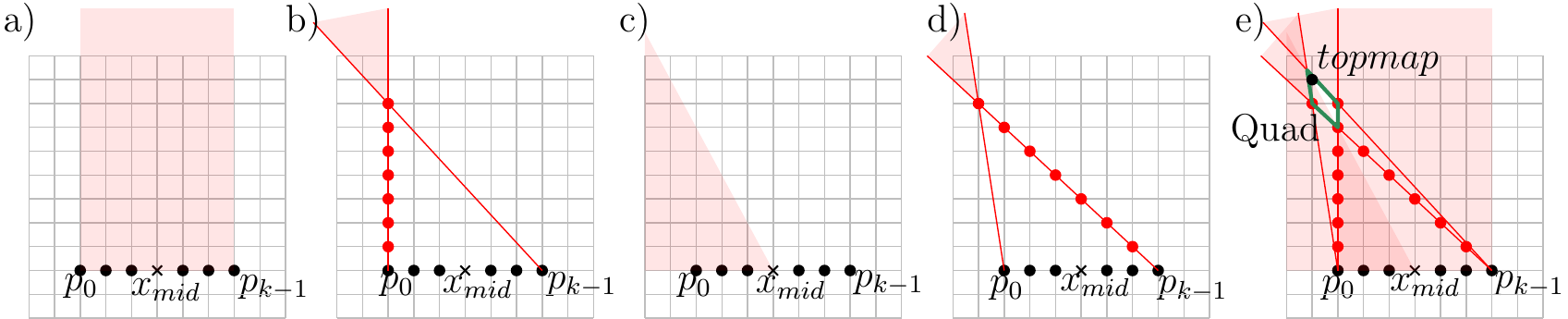}
    \end{center}
    \caption{a): $top'$ cannot lie within the red region. Otherwise $\medtriangleup'$ is 4-connected.
    b): $top'$ cannot lie within the red region. Otherwise the red line which contains more lattice point than the diameter would be in $\medtriangleup'$.
    c): $top'$ cannot lie within the red region. Otherwise there would be a horizontal shear mapping such that $top'$ is closer to $x_{mid}$ than it currently is.
    d): $top'$ cannot lie within the red region. Otherwise the red line which contains more lattice point than the diameter would be in $\medtriangleup'$.
    e): Superposition of the four previous figures. There is only one potential location for $top'$.}
    \centering
    \label{fig:topmap_pos}
\end{figure}


\subsection{Special case study} \label{ss:pom_pom}
In this section, we consider the situation where $top'$ is located at $(-1,k+1)$.
Notice that in this situation no lattice points can be added to $\medtriangleup'$ above $d_1$ without adding either the point $(-1,k)$ or the point $(0,k)$. Adding either of those points to $\medtriangleup'$ is impossible as $d_1$ would no longer be the lattice diameter (See figure~\ref{fig:pompom} a). This means that at this step, $\medtriangleup'$ is not 4-connected if and only if $\medtriangleup'$ is equal to the intersection of $\ZZ^2$ with the triangle $(0,0), (k-1,0), (-1,k+1)$. This also implies that the part of $S_1$ above $d$ is equal to $\medtriangleup'$.

We also notice that in this situation $\medtriangleup'$ is not only not 4-connected, but not even 8-connected.
However it is possible to apply to $\medtriangleup'$ a vertical shearing that maps $d_1$ to the line $y=x$. The resulting set that we call a \emph{pompom triangle} is almost 4-connected (See figure~\ref{fig:pompom} b).
Note that the pompom triangle is not unimodularly equivalent to any 4-connected set. The pompom triangle for $k=2$ consisting of only 4 points is the smallest digital convex set that is not unimodularly equivalent to any 4-connected set.

 \begin{figure}[tb] 
    \begin{center}
        \includegraphics[scale=0.7]{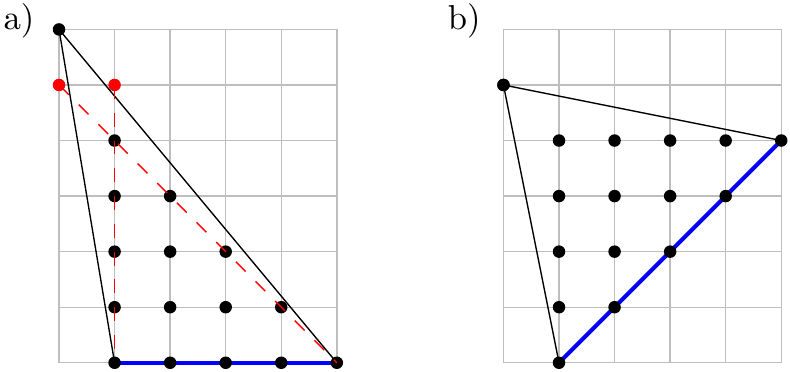}
    \end{center}
    \caption{a) The only possible location for $top'$ in which $\medtriangleup'$ is not 4-connected. As $d_1$ (shown in blue) is a lattice diameter neither of the red points can be in $\medtriangleup'$. As adding any points above $d$ would imply adding one of the red points, the set to which $S$ maps does not contain any other points above $d$ than those in $\medtriangleup'$ shown here in black.
    b) Representation of the same set as in a) after the vertical shear mapping making it almost 4-connected}
    \centering
    \label{fig:pompom}
\end{figure}

We now consider $S'$, the image of $S$ after all the mappings previously described in the case where the top part of the set is equal to the not 4-connected pompom triangle.
The position of $bottom$ after the mappings is denoted $bot'$. We consider the potential location of $bot'$.
\begin{itemize}
    \item $S'$ cannot contain any of the points $(k,0), (k,1),...(k,k-1)$. Otherwise $d_1$ would not be a lattice diameter. This defines $k-1$ cones in which $bot'$ cannot be located (See Figure~\ref{fig:pompom_bot} a).
    \item $S'$ cannot contain the point $(0,k-1)$. Otherwise $d_1$ would not be a lattice diameter. This defines a cone in which $bot'$ cannot be located (See Figure~\ref{fig:pompom_bot} b).
    \item $S'$ does not contain the point $(k,k)$. This defines a cone in which $bot'$ cannot be located (See Figure~\ref{fig:pompom_bot} c).
    \item The same constraints can also by applied using symmetry by reflection on the line $y = k-1-x$ (See Figure~\ref{fig:pompom_bot} e).
\end{itemize}

All those location constraint results in the $x$-coordinate of $bot'$ being between $0$ and $k-1$, which implies the bottom part of $S'$ located below $d$ is 4-connected, and hence $S'$ is almost 4-connected.

 \begin{figure}[tb] 
    \begin{center}
        \includegraphics[scale=0.592]{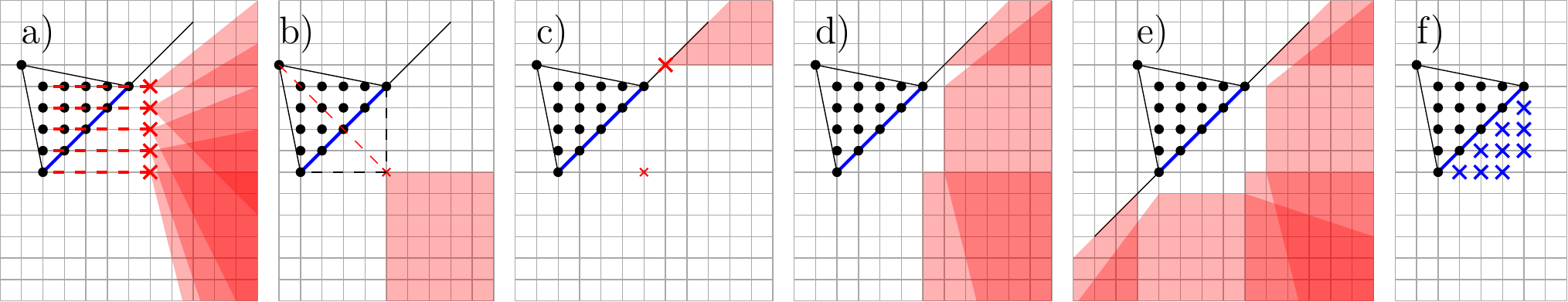}
    \end{center}
    \caption{a-e) Visual representation of the different location in which $bot'$ cannot be located when the top part is a pompom triangle. f) The remaining possible location for $bot'$ are all within a position that makes $\medtriangledown'$ 4-connected.}
    \centering
    \label{fig:pompom_bot}
\end{figure}


\subsection{Connecting the bottom} \label{ss:bottom_part}

We now consider the case where $\medtriangleup'$, the top part of the set $\meddiamond'$, which is the intersection of $\ZZ^2$ with the triangle $p_0(0,0), p_{k-1}(k-1,0), top'(a,b)$ is 4-connected, and we will focus on finding a horizontal shear mapping that almost 4-connects $\meddiamond'$.

As we moved $top'$ as close as possible to the middle of $p_0 p_{k-1}$ in Section~\ref{ss:top_part}, we have the following inequality: $|\frac{k-1}{2}-a| \leq \frac{b}{2}$.
We also showed that $0 \leq a \leq k-1$.

We say that a point $p(x,y)$ is \emph{directly above} (resp. \emph{directly below}) the line segment $d_1((0,0), (k-1,0))$ if $0 \leq x \leq k-1$ and $y \geq 0$ (resp. $y \leq 0$).

We now consider the possible locations for the bottommost point $bot'$.
If $bot'$, is directly below $d_1$, $\meddiamond'$ is trivially 4-connected, so we only consider the situation where $bot'$ is not directly below $d_1$.
We forget the assumption made in Section~\ref{ss:top_part} about the position of $top'$ relative to $x_{mid}$. 
However, we still keep the assumption that states that $top'$ is furthest or equally furthest than $bottom'$ from the line $y=0$.
For simplicity, we assume that $bot'(x_b,y_b)$ is located to the right of $d_1$, that is when the $x$-coordinate of $bot'$ is larger than $k-1$. The situation where $bot'$ is located to the left of $d_1$ is symmetrically identical. We will not consider it.

As we are not in the pompom case, we know that $top'$ is located directly above $d_1$, hence its $y$-coordinate is at most $k-1$ as otherwise $d_1$ would not be a lattice diameter. 

In order to prove that $\meddiamond'$ is 4-connected, we will now prove that there is a horizontal shear mapping that maps $\meddiamond'$ to $\meddiamond''$ such that
\begin{itemize}
    \item $top''$ remains directly above $d_1$, which guarantees the connectivity of $\medtriangleup ''$
    \item The lattice points $p_i(x,y) = bot" + i \overrightarrow{(-1,1)}$ such that $i > 0, y < 0$ are all within $\medtriangledown''$
    \item For either $v=\overrightarrow{(-1,0)}$ or $v=\overrightarrow{(0,1)}$, the following is true for all $p_i(x,y)$ such that $i \neq 0, y \neq 0$: $p_i + v \in \meddiamond$ (See Figure~\ref{fig:what_we_prove} a).
\end{itemize} 

\begin{figure}[tb] 
    \begin{center}
        \includegraphics[scale=0.67]{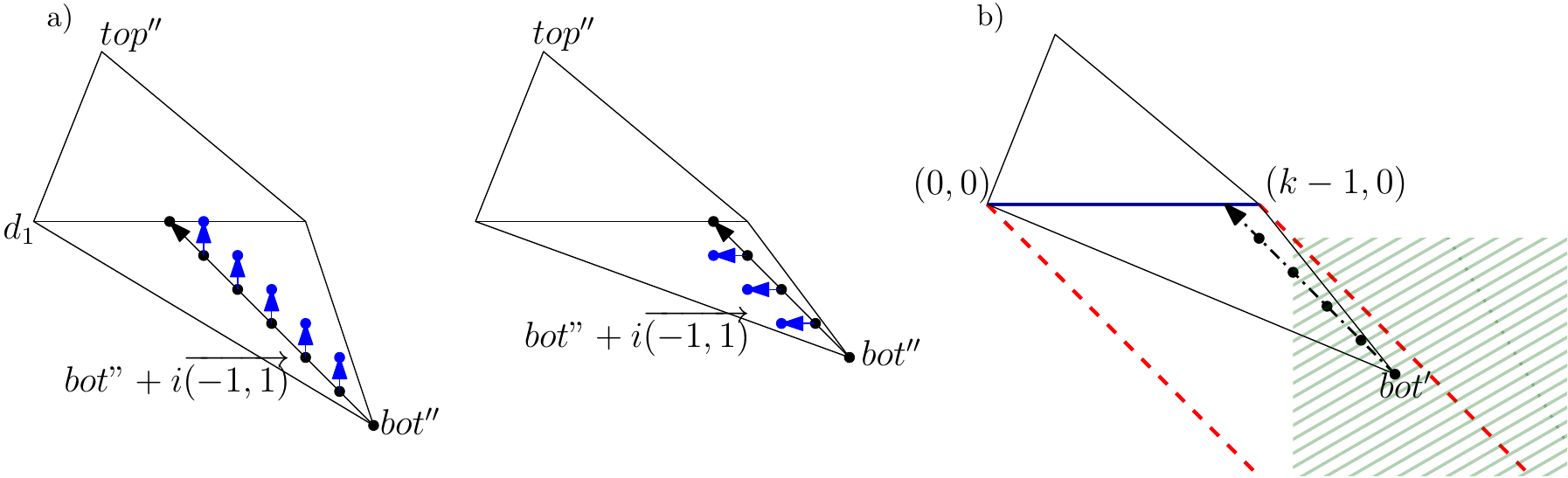}
        
    \end{center}
    \caption{a) The final set $\meddiamond''$ that we will obtain always contains the lattice points on the diagonal $bot" + i (-1,1)$ shown by the black arrow. $\meddiamond''$ will also contains, either the lattice points to the left or above the diagonal, shown in blue here.
    b) $bot'$ is contained in the green tiled region. As long as $bot'$ is located left of the line $y=k-1-x$, $\medtriangledown'$ contains the upper left diagonal from $bot'$, making $\medtriangledown'$ 8-connected.}
    \centering
    \label{fig:what_we_prove}
\end{figure}

In order for $\medtriangledown'$ to contain the lattice points $p_i(x,y)$ located on the diagonal $bot" + i \overrightarrow{(-1,1)}$ such that $i > 0$ and $y < 0$, $bot'(x_b,y_b)$ has to be located in between the two lines $y=-x$ and $y=k-1-x$. As we know that $x_b > k-x$, and $y_b \leq k-1$, this is equivalent to $bot'$ being located left of the line $y=k-1-x$ (See Figure~\ref{fig:what_we_prove} b).

First, we consider the situation where $bot'(x_b,y_b)=(k+\lambda,y_b)$ is in the region to the right of the line $y=k-1-x$, and we will show that there is an horizontal shear mapping that maps $bottom'$ to the left of the line $y=k-1-x$, and also maps $top'$ directly above $d_1$.
As $(k,0)$ is not in $\meddiamond'$, we know that $top'$ lies below the line supported by $(k,0)$ and $bot'$, that is the line $y = -\frac{y}{\lambda}+\frac{yk}{\lambda}$.
As $bot'$ is a lattice point right of the line $y=k-1-x$, $bot'$ is to the right of, or on, the line $y=k-x$. Hence, $y_b \geq \lambda$.
We call $m$ the integer such that $m < \frac{|y_b|}{\lambda} \leq m+1$.
The point $bot'$ is located in the wedge above $y = -\frac{1}{m}x+\frac{k}{m}$ and below $y = -\frac{1}{m+1}x+\frac{k}{m+1}$ (See Figure~\ref{fig:shift_from_right}).
We now apply to $\meddiamond'$ the horizontal shear mapping that maps $y = -\frac{1}{m}x+\frac{k}{m}$ to the vertical line $x=k$. This mapping, maps $y = -\frac{1}{m}x+\frac{k}{m}$ to the diagonal line $y = k - x$. Hence, after mapping $bot'$ is now located left of the line $y = k - x$.
As $top'$ is located left of the line $y=k-x$, after this mapping $top'$ still strictly lies to the left of the line $x=k$, and hence the triangle $(0,0), (k-1,0) top'$ is still 4-connected.

\begin{figure}[tb] 
    \begin{center}
        \includegraphics[scale=0.7]{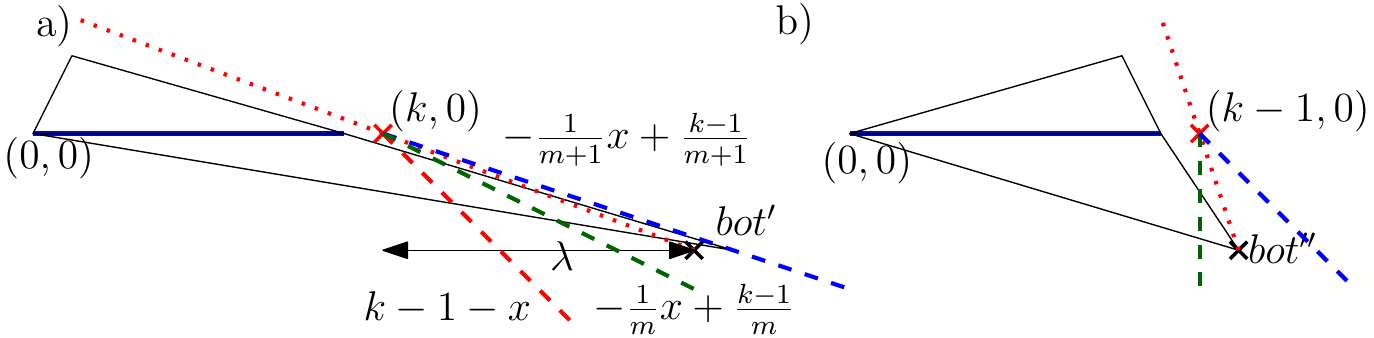}
    \end{center}
    \caption{a) $bot'$ is located in between the blue and red dashed lines, and $top'$ is located left to the red dotted line.
    b) Once the horizontal shear mapping is applied, $bot''$ is located left to the line $y=k-1-x$. As $top''$ is located left to the red dotted line, $top''$ is located left to the vertical line $x=k-1$. Hence $\medtriangleup''$ is 4-connected.}
    \centering
    \label{fig:shift_from_right}
\end{figure}

We will now consider the last case remaining, that is when $bot'$ is located left of, or on, the line $y=k-1-x$. To show that, in this situation, $\medtriangledown'$ is almost 4-connected we will use two lemmas. Lemma~\ref{l::left_connect} that explicits the location in which $\medtriangledown'$ contains the lattice points on the line $bot' + (-2,1) + i \overrightarrow{(-1,1)}$ (See, Figure~\ref{fig:contains_left} b)), and Lemma~\ref{l::top_connect} that explicits the location in which $\medtriangledown'$ contains the lattice points on the line $bot' + (-1,2) + i \overrightarrow{(-1,1)}$ (See, Figure~\ref{fig:contains_left} a).

\begin{lemma} \label{l::left_connect}
    Any triangle $\medtriangledown:(0,0), (l,0), p(x,y)$, with $l>0$, $y<0$ and $x>l>1$ is almost 4-connected when:
    $\frac{-x}{2} \leq y \leq l-x$, 
\end{lemma}
\begin{proof}
    We first consider the case where the point $p(l-y,y)$ is located on the line $y=l-x$. The right edge of $\medtriangledown$ is supported by the line $y=l-x$, and hence contains a lattice point for each integer $y$-coordinate the edge crosses. The horizontal width of $\medtriangledown$ at the $y$-coordinate $y+1$ is equal to $\frac{l}{y} \geq 1$. Hence both the points $(l-y-1,y+1)$ and $(l-y-2,y+1)$ are in $\medtriangledown$ and, the width of the triangle increasing with $y$-coordinate, the same reasoning can be done with the other points on the diagonal $p + k \overrightarrow{(-1,1)}$. \\
    We now study, how far to the left of the line $y=l-x$ can $p(x,y)$ move, such that $\medtriangledown$ still contains $(x-1,y+1)$ and $(x-2,y+1)$.
    We rewrite the coordinates of $p$ in the following manner: $p(l-y-\lambda, y)$.
    The left edge of the triangle is supported by the line $y = \frac{y}{l-y-\lambda}$.
    At height $y+1$ the left edge $x$-coordinate $x_l$ is equal to: $x_l = \frac{(y+1)(l-y-\lambda)}{y} = l-\lambda-1-y+\frac{l-\lambda}{y}$. \\
    As we want, $(l-y-\lambda-2, y+1)$ to be inside $\medtriangledown$, we need\\ $x_l \leq l-y-\lambda-2$. That is \\ $l-\lambda-1-y+\frac{l-\lambda}{y} < l-y-\lambda-2$ \\
    $\lambda \leq y+l$ \\
    As $\medtriangledown$ contains all points on the diagonal $p + k \overrightarrow{(-1,1)}$, $\medtriangledown$ is almost 4-connected when $0 \leq \lambda \leq y+l$. That means, considering $p(x,y)$, $\medtriangledown$ is connected if $-2y \leq x \leq l-y$. Which is equivalent to $\frac{-x}{2} \leq y \leq l-x$.
\end{proof}

\begin{lemma} \label{l::top_connect}
    The triangle $\medtriangledown:(0,0), (l,0), p(x,y)$, with $l>0$, $-l<y<0$ and $x>l>1$ is almost 4-connected when:
    $y \leq 2l-2x$.
\end{lemma}
\begin{proof}
    We first consider the case where the point $p(l,y)$ is located on the line $x=l$.
    In this situation, it is clear that $(l-1,y+2)$ is in $\medtriangledown$. Now we want to know, how far to the right of the line $x=l$ can $p(x,y)$ be located such that $(x-1,y+2)$ is in $\medtriangledown$.
    We rewrite the coordinates of $p$ in the following manner: $p(l+\lambda, y)$.
    The right edge of $\medtriangledown$ is located on the line $y=\frac{y}{\lambda}x-\frac{yl}{\lambda}$. As we want $(x-1,y+2)$ in $\medtriangledown$, we want the right edge to intersect the line $ y = l+\lambda-1$ above $y+2$. That is \\
    $\frac{y}{\lambda}(l+\lambda-1)-\frac{yl}{\lambda} \geq y+2$ \\
    $\frac{y(\lambda-1)}{\lambda} \geq y+2$ \\
    $\lambda \leq \frac{-y}{2}$ \\
    That means, considering $p(x,y)$, $\medtriangledown$ is connected when $ x \leq l - \frac{y}{2}$. Which is the same as $y \leq 2l-2x$.
\end{proof}

\begin{figure}[tb] 
    \begin{center}
         \includegraphics[scale=0.629]{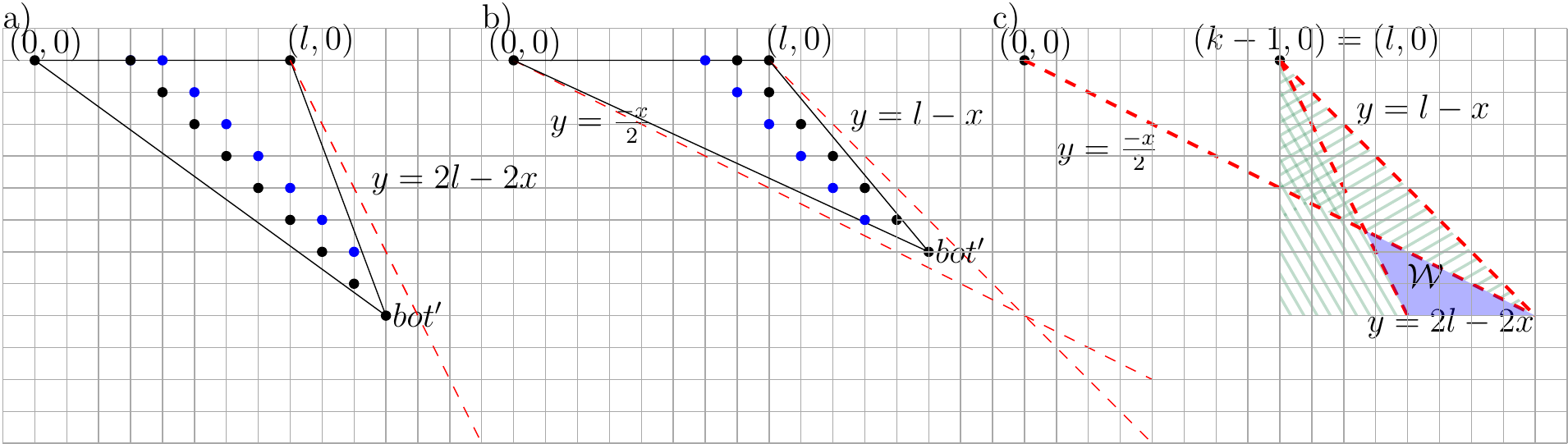}
    \end{center}
    \caption{a) If $bot'$ is located left of the dashed red line, then $\medtriangledown'$ contains the black diagonal, and the blue one just above it.
    b) If $bot'$ is located in the top wedge in between the two dashed red line, then $\medtriangledown'$ contains the black diagonal, and the blue one just left of it. c) If $bot'$ is located in one of the green tiled surfaces, then $\medtriangledown'$ is almost 4-connected. That is not true when $bot'$ is located in the blue wedge $\mathcal{W}$}
    \centering
    \label{fig:contains_left}
\end{figure}

We now consider the last surface in which $bot'$ can be located such that $\medtriangledown'$ is not almost 4-connected. That is the wedge $\mathcal{W}$ located below $y=-\frac{x}{2}$ and above $y=2k-2-2x$ (See figure~\ref{fig:contains_left} c).

To study the situation where $bot'$ is located inside $\mathcal{W}$ we have to consider multiple possible location for $top'$. If $top'$ is located below the line $y = k-1-x$, then we can apply a horizontal shearing that maps $\mathcal{W}$ to a surface that is directly below $d_1$, and maps $\medtriangleup'$ to a 4 connected set. If $top'$ is located above the line $y = k-1-x$ and above the line $y=x$, then the lattice points inside the triangle $\medtriangleup_{top} (0,0),(k-1,0), (\frac{k-1}{2},\frac{k-1}{2})$ are in $\meddiamond'$, and since $\medtriangleup_{top} + k \overrightarrow{(1,-1})$ covers $\mathcal{W}$, $bot'$ cannot be located inside $\mathcal{W}$ with $d_1$ being a lattice diameter.
Finally, $top'$ cannot be located below the line $y=x$ as the fact that $top'$ is the furthest point from the line $y=0$ and that the point $(k,0)$ is not in $\meddiamond'$ makes it impossible for $top'$ to be located below the line $y=x$.
A more detailed proof of this affirmation is available in appendix~\ref{app:un}.
With $\mathcal{W}$ covered, we now have considered all possibilities for the location of $bot'$ and found in each case a mapping that maps $\meddiamond'$ to an almost 4-connected set. Furthermore, in the event where $\meddiamond'$ is not mapped to 4-connected, the point that is not 4-connected to the set is $bot''$ that image of $bot'$.

\subsection{Back to the general case} \label{ss:general_case}

We now show how to obtain a mapping that maps $S$ to an almost 4-connected set from the mapping $\mathcal{M}$ we previously described that maps $\meddiamond$ to $\meddiamond''$, an almost 4-connected set. Using Lemma~\ref{le:quadri} we can discard the case where $\meddiamond''$ is 4-connected, and using the same arguments as in the proof of Lemma~\ref{le:quadri} we can conclude that the only points that might not be 4-connected in $S''$ are the one that have the same $y$-coordinate as $bot''$.
As $bot''$ is not 4-connected, it means that $bot''$ is located right of the line $x = k-1$, in this situation we showed that $\meddiamond''$ contains the lattice point $bot'' + (-1,1)$. As a consequence, adding a point to the left of $bot''$ would make $S''$ 4-connected.
Now, if we add the point $p_a = bot'' + (1,0)$, either $p_a$ is below or on the line $y = k-1-x$, and hence $S''$ contains $p_a  + (-1,1)$, which is the point above $bot''$ which makes $S''$ 4-connected, or $p_a$ is located exactly on the line $y = k-x$. As a consequence, since $S''$ does not contain $(k,0)$, $top''$ is located left to the line $y=k-x$ which means we can apply to $S''$ the horizontal shearing that maps $y=k-x$ to $x=k$ in order to obtain a 4-connected set.

\section{Algorithm}
In this section we study the algorithmic complexity of finding an almost 4-connected unimodularly equivalent set to a digital convex set $S$ of $n$ points and of diameter $r$.
We propose an algorithm that runs in $O(n + h \log r) = O(n + n^\frac{1}{3} \log r)$ time, where $h$ is the number of vertices on the convex hull of $S$. 
This algorithm mimics the construction done in section~\ref{s:eq}, and hence relies on the computation of a lattice diameter of $S$.

\subsection{Computing the lattice diameter}

We present here an algorithm to compute a lattice diameter of a digital convex set $S$ in $O(n + h \log r)$ time, or $O(h\sqrt{n} + h \log r)$ when the convex hull of $S$ is known. This algorithm relies on the fact that at least one of the vertices of the convex hull of $S$ is located on a lattice diameter~\cite{BaP92,BaF01}.
Hence, we only have to consider vertices of the convex hull to find a lattice diameter.
We consider $v$, a vertex of $\conv(S)$ such that $v$ is on a lattice diameter $d$ of $S$. We now consider the fan triangulation $\mathcal{T}$ rooted on $v$ of $\conv(S)$. That is the triangulation defined by all the diagonals going from the vertex $v$ to all of the other vertices of $\conv(S)$.
There is a triangle $t_1$ in $\mathcal{T}$ that contains $d$ (See Figure~\ref{fig:triangle_diameter} a). The base of $t_1$ is an edge of $\conv(S)$, and its opposite vertex is $v$.
Hence, in order to find a lattice diameter of $S$ we can test each triangle of each of the $h$ fan triangulations. In each of those triangles $t_1(v_1,v_2,v_3)$, we want to compute a line $\ell$ such that $v_1$ is located on $\ell$, and $\ell$ maximizes $|\ell \cap t_1 \cap \ZZ^2|$.
In order to compute such a line, we apply to $t_1$ a unimodular affine transformation that maps the line supported by $v_2 v_3$ to a horizontal line (See Figure~\ref{fig:triangle_diameter} b).
The three vertices of $t_1$: $v_1, v_2$ and $v_3$ are now mapped to $v'_1(x'_1,y'_1)$, $v'_2(0,0)$ and $v'_3(x'_3,0)$, the three vertices of $t'_1$.
For simplicity, we assume that $v'_1$ is above $v'_2$ and $v'_3$.
We now consider any line $\ell$ that goes through $v'_1$ and intersects the segment $v'_2 v'_3$, more specifically we consider the line segment supported by $\ell$ that is located inside $t'_1$ $\ell_s = \ell \cap t'_1$.
As we mapped $v'_2 v'_3$ to an horizontal line, the number of horizontal lines $y = i, i \in \ZZ$ intersected by $\ell_s$ is always equal to the constant number $y'_1 - y'_3 + 1$.
As lattice points on a line are evenly separated, and since $v'_1$ is a lattice point, maximizing the number of lattice points on $\ell_s$ is equivalent to finding $l_t(x_t,y_t) \neq v'_1$, the top most lattice point in $t'_1$. The number of lattice points in $t'_1$ on the line $v'_1 l_t$ is equal to $\floor{\frac{y'_1}{y'_1-y_t}}+1$ (See Figure~\ref{fig:triangle_diameter} b).

\begin{figure}[tb] 
    \begin{center}
        \includegraphics[scale=0.7]{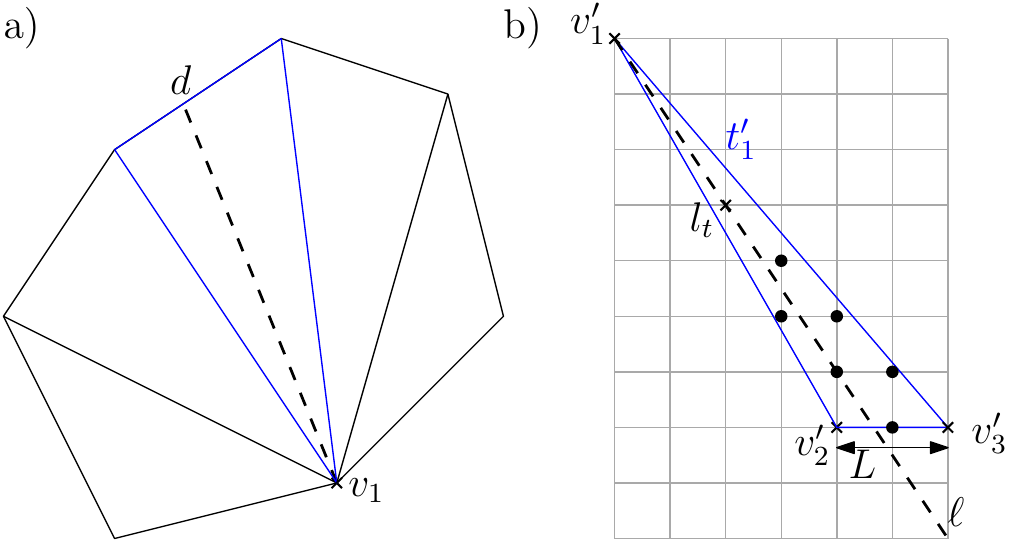}
    \end{center}
    \caption{a) The lattice diameter $d$ that goes through $v_1$ is in one of the triangles of the fan triangulation of $S$ from $v_1$. This triangle is shown in blue.
    b) The triangle is mapped to a triangle where $v'_2 v'_3$ is a horizontal line. In this triangle, finding $\ell$ (the longest digital segment in $t'_1$ going through $v'_1$) is equivalent to finding $l_t$ the topmost lattice point in $t'_1 \cap \ZZ^2 \setminus v'_1$.}
    \centering
    \label{fig:triangle_diameter}
\end{figure}

Naively looking for $l_t$ by testing the intersection of all horizontal lines with $t'_1$ starting from $y'_1$ to $0$ leads to at most $y'_1 + 1$ tests.
We now show that $y'_1 \leq 2n$. Indeed, the area of $t'_1$ is equal to $\frac{(y'_1)x'_3}{2}$, and we now that $x'_3 \geq 1$.
Using Pick's formula~\cite{Pic99}, that states that in a lattice triangle $\triangle$ we have the following equality: $A = i + \frac{b}{2} - 1$ where $A$ is the area of $\triangle$, $i$ is the number of lattice points strictly inside $\triangle$ and $b$ is the number of lattice points on the edges of $\triangle$, we can deduce that the number of lattice points inside $t'_1$ $n_t = t'_1 \cap \ZZ^2$ is at least equal to $\frac{y'_1}{2}$.
This means that $l_t$ can be found in $O(n_t)$ time.
However, using a lower bound $\lambda$ on $n_d$ the number of lattice points on a lattice diameter of $S$, we can stop computation in $t'_1$ before finding $l_t$ after $\frac{y'_1}{\lambda - 1}$ steps. Indeed, any below that would be too far from $v'_1$ to result in a lattice diameter of $S$ as the segment would contain less than $\dfrac{y'_1}{\frac{y'_1}{\lambda - 1}}+1 = \lambda$ lattice points.
We use a very rough lower bound to $n_d$ that can be directly deduced from the construction we made in order to prove Theorem~\ref{th:1} by fattening the bounding box of $\meddiamond''$ by $n_d$ both to the left and to the right. The resulting bound is the following: $\frac{1}{8} \sqrt{n} \leq n_d$. 

We now consider the total number of operations needed to compute a lattice diameter of $S$.
For a given vertex  $v$ of $\conv(S)$, we denote $(n_1, n_2, ..., n_i, ... n_h)$ the number of lattice points inside each of the $h$ triangles of the fan triangulation rooted on $v$.
As $v$ is in all the $h$ triangles, and any other lattice point can only be in at most 2 triangles, $\displaystyle \sum_i n_i \leq 2n + h \leq 3n$.
In addition to the computation of all the $h$ unimodular affine transformations required in order to map the edges of $\conv(S)$ to horizontal lines, computing a potential lattice diameter that goes through a given vertex $v$ takes at most, $\displaystyle \sum_i \dfrac{2 n_i}{\frac{1}{8}\sqrt{n}-1} \leq \dfrac{6n}{\frac{1}{8}\sqrt{n}-1} = O(\sqrt{n})$ time.
Repeating this process for all the $h$ vertices of $\conv(S)$ takes at most $O(h \sqrt{n})$ time.
Now, adding to this the computing time required to compute the $h$ unimodular affine transformations that map each edge of $\conv(S)$ to a horizontal line, we obtain a time complexity of $O(h \sqrt{n} + h \log r)$ to compute the lattice diameter of $S$, given its convex hull.

Computing the convex hull of a digital convex set can be done in linear time using the quickhull algorithm~\cite{CFG19}, and since there is at most $O(n^{1/3})$ vertices on $\conv(S)$~\cite{Zun98,BaI08} the total time complexity of the algorithm in order to compute a lattice diameter of a digital convex set is $O(n + h \log r) = O(n +n^\frac{1}{3} \log r)$.

\subsection{Computing a unimodularly equivalent almost 4-connected set}

Once the lattice diameter computed, the computation of an affine isomorphism of $\ZZ^2$ resulting in an almost 4-connected set requires:
\begin{itemize}
    \item The computation of the affine isomorphism mapping $d$ to a horizontal line in $O(\log r)$ time, where $r$ is the diameter of $S$.
    \item Applying the mapping to $conv(S)$ in $O(h)$ time.
    \item Computing $\meddiamond$ in $O(h)$ time.
    \item Computing the horizontal shear mapping positioning the topmost point in $O(1)$ time.
    \item Computing a horizontal shear mapping from the positions of the top most and bottom most points in order to make $\meddiamond'$ 4-connected in $O(1)$ time.
    \item applying all the mappings to $conv(S)$ in $O(h)$ time.
    \item Eventually computing one last shear mapping in the event where $S$ does not map to an almost 4-connected set in $O(h)$.
\end{itemize}

Hence, the total time complexity sums to $O(n)$. Adding the time complexity of the lattice diameter algorithm results to a time complexity of $O(n + h \log r)$ time in order to compute a unimodularly equivalent almost 4-connected set.

\section{Perspective}
In this paper, in 2 dimension, we proved for every digital convex set the existence of a unimodularly equivalent almost 4-connected set.
While it is proven that an infinite amount of digital convex sets cannot be mapped to unimodularly equivalent 4-connected sets, the algorithm proposed in this paper does not ensure to provide a 4-connected unimodularly equivalent set when such a set exists.

\bibliographystyle{unsrt}
\bibliography{sample}

\newpage

\begin{subappendices}
\renewcommand{\thesection}{\Alph{section}}%

\section{Study of the case where $bot'$ is located in the wedge $\mathcal{W}$}\label{app:un}

Here, we consider what happens when $bot'$ is located in the wedge $\mathcal{W}$ defined by the surface below $y=-\frac{x}{2}$ and above $y=2k-2-2x$.
We consider 3 possible location for $top'$.
\begin{itemize}
    \item below or on the line $y = k-1-x$
    \item above or on both lines $y= k-1-x$ and $y = x$
    \item strictly below the line $y = x$
\end{itemize}

\paragraph{Case 1: $top'$ located below $y = k-1-x$\\}
In this situation the horizontal shear mapping that transforms $y = k-1-x$ to the line $y = k-1$ moves $bot'$ directly below $d$ while $top'$ remains directly above $d$. After this shear mapping, the set is hence 4-connected (See Figure~\ref{fig:wedge_left}).

\begin{figure}[tb] 
    \begin{center}
        \includegraphics[scale=0.7]{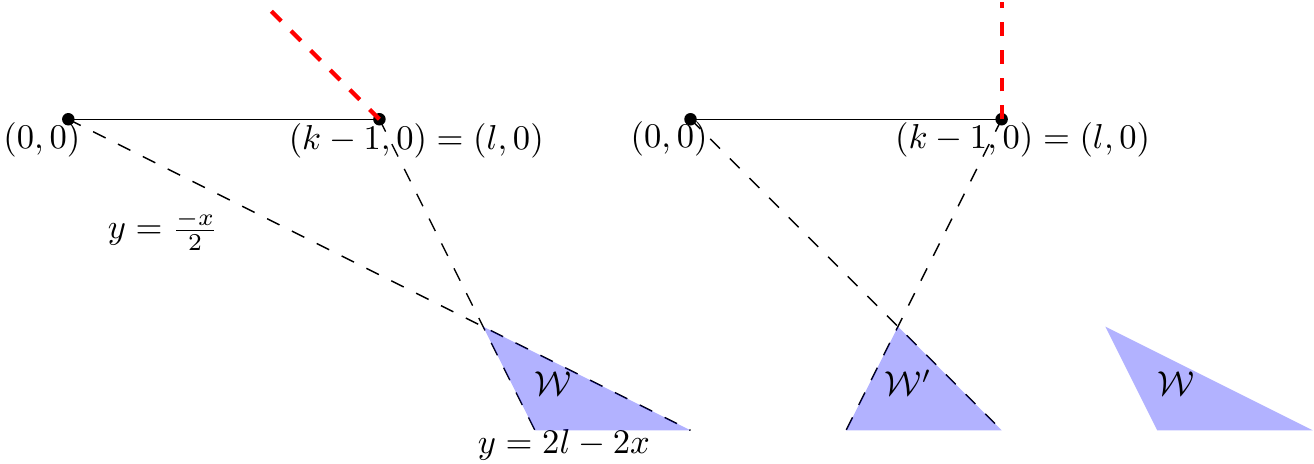}
    \end{center}
    \caption{$top'$ being below the red dashed line, we can apply a horizontal shear mapping to $\meddiamond$ such that the wedge $\mathcal{W}$ maps entirely directly below the lattice diameter $d_1$. This makes $\meddiamond$ 4-connected.}
    \centering
    \label{fig:wedge_left}
\end{figure}

\paragraph{Case 2: $top'$ located above or on $y= k-1-x$ and $y = x$\\}
In this situation, all the lattice points inside the triangle $\triangle_{top} (0,0), (k-1,0), (\frac{k-1}{2},\frac{k-1}{2})$ are part of $\meddiamond'$ (See Figure~\ref{fig:wedge_top} a).
Since $\meddiamond'$ is digitally convex, for any points $p_i$ in $\triangle_{top} \cap \ZZ^2$, the point $p_i + k \overrightarrow{(1,-1)}$ cannot be in $\meddiamond'$, else $d$ would no longer a lattice diameter (See Figure~\ref{fig:wedge_top} b).
As a consequence, $bot'$ cannot be located in the triangle $\medtriangledown'(k,-k), (2k-1,-k), (\frac{3k-1}{2},\frac{-k-1}{2})$.

\begin{figure}[tb] 
    \begin{center}
        \includegraphics[scale=0.7]{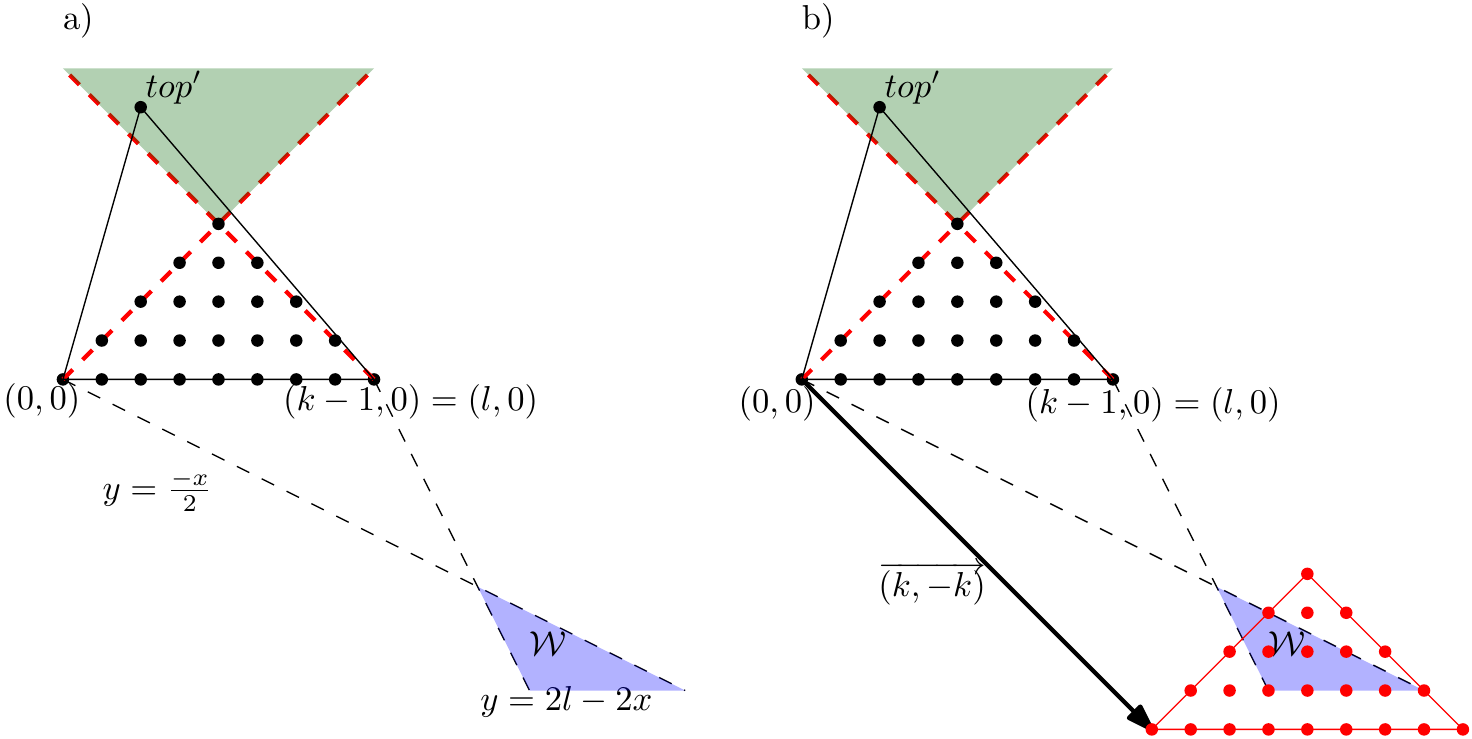}
    \end{center}
    \caption{a) $top'$ being located in the green wedge, all the black lattice points are in $\medtriangleup'$. 
    b) The lattice diameter of $\meddiamond$ containing only $k$ points, none of the red points can be located in $\medtriangledown'$.}
    \centering
    \label{fig:wedge_top}
\end{figure}

This leaves a triangle in which $bot'$ can still be located. This triangle is; above the line $y = x - 2k$, above the line $y = 2k-2-x$ and below the line $y = \frac{-x}{2}$.
We will now show that there is no lattice point located inside this triangle, and that hence $bot'$ cannot be located inside it.
The three vertices of the triangle are: 
$(\frac{4k-4}{3},\frac{-2k+2}{3})$, 
$(\frac{4k}{3},\frac{-2k}{3})$, 
$(\frac{4k-2}{3},\frac{-2k-2}{3})$.
We consider all the lattice points inside the square $x(,y) \in ]\frac{4k-4}{3},\frac{4k}{3}[ \times ]\frac{-2k-2}{3},\frac{-2k+2}{3}[$
There can only be three lattice points inside this square depending on the value of $k \pmod 3$. Those points are:
$(\frac{4k-3}{3},\frac{-2k}{3})$ which is located on the line $y=2k-2-2x$,
$(\frac{4k-2}{3},\frac{-2k+1}{3})$ which is located on the line $y=\frac{-x}{2}$, and
$(\frac{4k-1}{3},\frac{-2k-1}{3})$ which is located on the line $y=x-2k$.
This means, that if $top'$ is located above or on the lines $y = x$ and $y = k-1-x$, then $bot'$ cannot be located inside the wedge $\mathcal{W}$.

\paragraph{Case 3: $top'$ located strictly below $y = x$\\}
We now show, that $top'$ cannot lie to the right of the line $y=x$.
For simplicity, we denote $l=k-1$.
The intersection point of the two lines defining the wedge $\mathcal{W}$ is: $(t(\frac{4}{3}l,-\frac{2}{3}l))$.

We consider three different cases, either $l \equiv 0 \pmod 3$, or, $l \equiv 1 \pmod 3$, or $l \equiv 2 \pmod 3$.

\paragraph{Case 3.1: $l \equiv 0 \pmod 3$}
Amongst all points located strictly within the wedge $\mathcal{W}$ the topmost point is $t_w(\frac{4}{3}l+1,-\frac{2}{3}l-1)$ (See Figure~\ref{fig:wedge_right} b).
This means that $top'$, is above or on the line $y = \frac{2}{3}l+1$ (See Figure~\ref{fig:wedge_right} a).
We now look for the lattice point $p_l$ strictly within $\mathcal{W}$ such that all other lattice points in the wedge lies right to the line $\ell$ supported by $(k,0)$ and $p_l$.
As $(k,0)$ is not in $\meddiamond'$, we know that $top'$ strictly lies to the left of  $\ell$. 
This point is $p_l(\frac{4}{3}l+1,-\frac{2}{3}l-1)$ (See Figure~\ref{fig:wedge_right} b).
The equation of $ell$ is $y = (-2-\frac{3}{l})x+(2+\frac{3}{l})(l+1)$.
$\ell$ intersects the line $y = \frac{2}{3}l+1$ at the point $'\frac{2}{3}l+1,\frac{2}{3}l+1)$. Hence, when $l \equiv 0 \pmod 3$ no lattice point lies strictly left to $\ell$, above or on $y = \frac{2}{3}l+1$, and strictly right to $y = x$, and hence $top'$ cannot be located below the line $y = x$.

\begin{figure}[tb] 
    \begin{center}
        \includegraphics[scale=0.645]{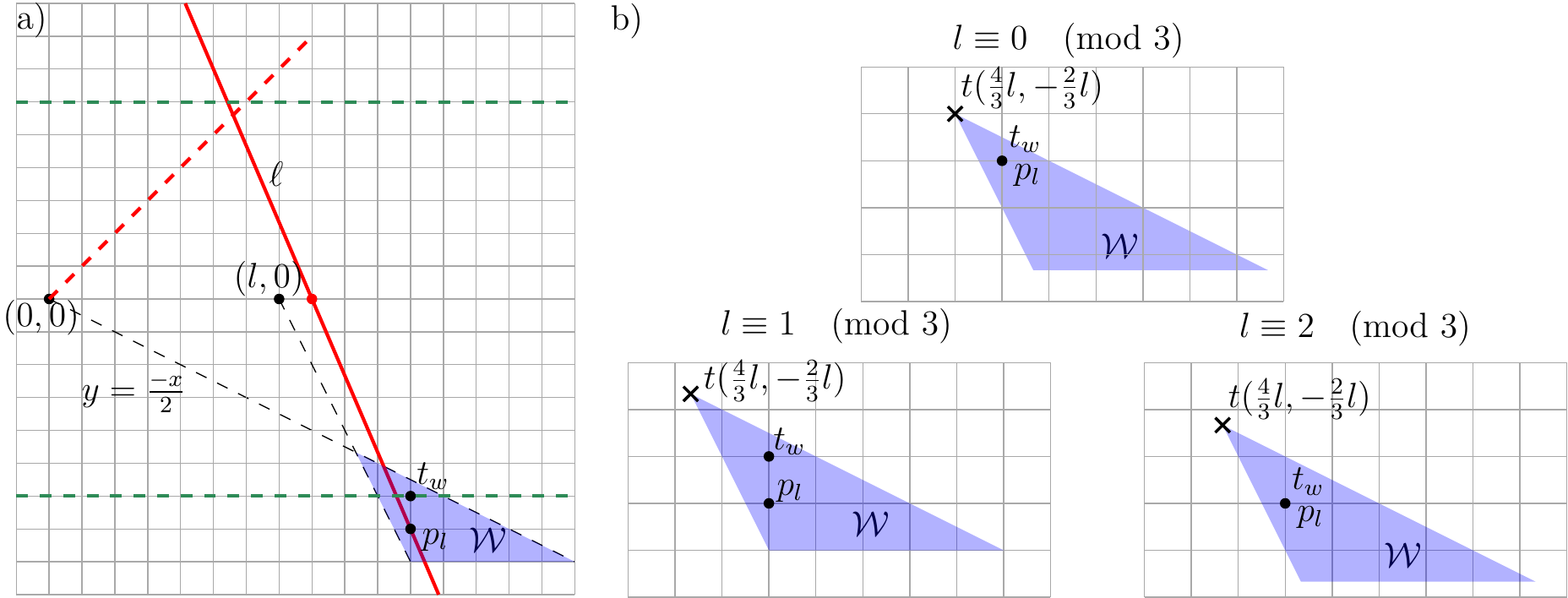}
    \end{center}
    \caption{a) As $bot'$ is located below the bottom green dashed line, $top'$ is located above the top green dashed line. As the lattice points in $\mathcal{W}$ are located right to $\ell$, ans as $\meddiamond'$ does not contain the red lattice point $(l+1,0)$, then $top'$ has to be located right of $\ell$.
    No lattice points are located, above $y=x$, to the left of $\ell$ and above the green dashed line. This situation is impossible.
    b) The location of $t_w$ (the topmost lattice point strictly within the wedge $\mathcal{W}$) and $p_l$ (the lattice point such that $\mathcal{W} \cap \ZZ^2$ lies right to the line $(l+1,0) p_l$)}
    \centering
    \label{fig:wedge_right}
\end{figure}

\paragraph{Case 3.2: $l \equiv 1 \pmod 3$}
Amongst all points located strictly within the wedge $\mathcal{W}$ the topmost point is $t_w(\frac{4l+5}{3},\frac{-2l-4}{3})$ (See Figure~\ref{fig:wedge_right} b).
This means that $top'$, is above or on the line $y = \frac{2l+4}{3}$.
We now look for the lattice point $p_l$ strictly within $\mathcal{W}$ such that all other lattice points in the wedge lies right to the line $\ell$ supported by $(k,0)$ and $p_l$.
As $(k,0)$ is not in $\meddiamond'$, we know that $top'$ strictly lies to the left of  $\ell$. (See Figure~\ref{fig:wedge_right} a).
This point is $p_l(\frac{4l+5}{3},\frac{-2l-7}{3})$ (See Figure~\ref{fig:wedge_right} b).
We now study the intersection point of $\ell$ with the line $x = \frac{2l+4}{3}$. If this intersection point is below the $y$-coordinate $\frac{2l+4}{3}$ then we know that no lattice point lies strictly left to $\ell$, above or on $y = \frac{2l+4}{3}$, and strictly right to $y = x$.
$\ell$'s equation is $y = \frac{2l+7}{l+2}(l+1-x)$. At $x=\frac{2l+4}{3}$ that is $\frac{(l-1)(2l+7)}{3l+6}$. We now compare that to $\frac{2l+4}{3}$: \\
$\frac{(l-1)(2l+7)}{3l+6} - \frac{2l+4}{3}$ =
$\frac{(l-1)(2l+7)}{3l+6} - \frac{(2l+4)(l+2)}{3(l+2)}$ =
$\frac{-l-1}{3l+6} < 0$.
Which proves that if $l \equiv 1 \pmod 3$ $top'$ cannot be located below the line $y = x$.

\paragraph{Case 3.3: $l \equiv 2 \pmod 3$\\}
Amongst all points located strictly within the wedge $\mathcal{W}$ the topmost point is $\mathcal{W}$ is $t_w(\frac{4l+4}{3},\frac{-2l-5}{3})$ (See Figure~\ref{fig:wedge_right} b).
This means that $top'$, is above or on the line $y = \frac{2l+5}{3}$.
We now look for the lattice point $p_l$ strictly within $\mathcal{W}$ such that all other lattice points in the wedge lies right to the line $\ell$ supported by $(k,0)$ and $p_l$.
As $(k,0)$ is not in $\meddiamond'$, we know that $top'$ strictly lies to the left of  $\ell$. (See Figure~\ref{fig:wedge_right} a).
This point is $p_l(\frac{4l+4}{3},\frac{-2l-5}{3})$ (See Figure~\ref{fig:wedge_right} b).

We now study the intersection point of $\ell$ with the line $x = \frac{2l+5}{3}$. If this intersection point is below the $y$-coordinate $\frac{2l+5}{3}$ then we know that no lattice point lies strictly left to $\ell$, above or on $y = \frac{2l+5}{3}$, and strictly right to $y = x$.
$\ell$'s equation is $y = \frac{2l+5}{l+1}(l+1-x)$. At $x=\frac{2l+5}{3}$ that is $\frac{(2l+5)(l-2)}{3(l+1)}$. We now compare that to $\frac{2l+5}{3}$: \\
$\frac{(2l+5)(l-2)}{3(l+1)} - \frac{2l+5}{3}$ =
$\frac{(2l+5)(l-2)}{3(l+1)} - \frac{(2l+5)(l+1)}{3(l+1)} < 0$

Which proves that if $l \equiv 2 \pmod 3$ $top'$ cannot be located below the line $y = x$.

\end{subappendices}

\end{document}